\documentclass[journal]{IEEEtran}
\usepackage{amsmath}
\usepackage{amssymb}
\usepackage{amsthm}
\usepackage{bm}
\usepackage{graphicx}
\usepackage{xcolor}

\newtheorem{theorem}{Theorem}
\newtheorem{corollary}{Corollary}

\newtheorem{proposition}{Proposition}
\newtheorem{lemma}{Lemma}

\begin{document}

\title{Orbital Plane Geometry and Information Conditioning
for Doppler-Only LEO Positioning}

\author{Charles E. Thornton \\ Virginia Tech NSI, Blacksburg, VA 
\thanks{Correspondence: cthorn14@vt.edu}}

\maketitle

\begin{abstract}
We study an idealized information model for Doppler-only
positioning with low earth orbit (LEO) signals of opportunity
from a stationary receiver. Motivated by the observation that Doppler measurements from a satellite pass provide information primarily within the
associated orbital plane, we model each satellite contribution as a weighted projection onto that plane. Under this model, the combined information matrix from multiple satellites is a sum of orbital-plane projection operators. An elementary lemma grounds the model in Keplerian dynamics by establishing near rank-deficiency of the pass-integrated Fisher information along the orbital-plane normal. Closed-form expressions are derived for the eigenvalues, condition number, and worst-case Cram\'er--Rao lower bound. For two satellites, the conditioning is governed by the dihedral angle between orbital planes and the relative information strengths of the
two links. Monte Carlo evaluation of pass-integrated Doppler Fisher information matrices demonstrates that the proposed surrogate captures the dominant conditioning trends associated with orbital-plane diversity, providing a simple geometric framework for understanding the impact of constellation geometry in Doppler-only positioning systems.
\end{abstract}

\section{Introduction}
Positioning using signals of opportunity (SoOP) from low earth orbit (LEO) satellites has attracted significant recent interest as a complement or alternative to GNSS in denied or degraded environments~\cite{kassas2019,psiaki2021}. In opportunistic settings, Doppler shift is often the primary navigation observable since time-of-arrival measurements require either transmitter cooperation or an independent timing reference. Consequently, Doppler-only positioning has emerged as a practical approach for exploiting signals from non-cooperative LEO systems such as Starlink, Orbcomm, and other communications constellations.

Existing work has demonstrated the feasibility of Doppler-based positioning and developed estimation algorithms for LEO signals~\cite{kassas2019,psiaki2021,khalife2019}, with early studies recognizing the role of Doppler contours and geometric dilution of precision~\cite{levanon1998}. More recently, GDOP-based analyses have examined constellation geometry for modern LEO systems~\cite{moralesferre2020,mclemore2021}, and Fisher-information-based analyses have characterized fundamental localization limits~\cite{emenonye2025}. Ground-station-centric statistical characterizations of LEO Doppler behavior have also recently appeared~\cite{luxu2026}, highlighting the role of station latitude and elevation constraints in the observable orbital geometry. However, the resulting information matrices are tied to detailed signal and geometry models and resist geometric interpretation. In practice, the information contributed by a satellite pass is highly anisotropic and depends strongly on orbital geometry, yet the geometric structure of Doppler-only positioning information remains poorly understood. 

In Doppler-only positioning, navigation information is extracted from the temporal evolution of Doppler measurements over a pass~\cite{psiaki2021,neinavaie2021}, and the pass-integrated Fisher information matrix (FIM) provides a natural summary of that aggregate information. Understanding its geometric structure is therefore a key step toward characterizing the conditioning and accuracy limits of Doppler-only systems. In this letter, we propose an idealized geometric model for Doppler-only LEO positioning motivated by the observation that pass-integrated Doppler information is concentrated primarily within the associated orbital plane. Each satellite pass is modeled as a weighted projection onto that plane, yielding a tractable information matrix whose eigenstructure admits closed-form analysis. Throughout this work we assume a stationary receiver; the extension to low-dynamics platforms is left for future investigation. 

\emph{Contributions:} The main contributions of this letter are as follows: 
\begin{enumerate}
    \item We provide an elementary lemma (Lemma~1) that establishes the orbital-plane rank-deficiency of the pass-integrated Doppler FIM directly from Keplerian dynamics and characterizes the conditions under which the projection approximation is valid.
    \item We present closed-form expressions for the eigenvalues, condition number, and worst-case CRLB of the resulting information model (Theorem~1, Corollary~1), extended to $N$ satellites via a pairwise conditioning bound (Proposition~1).
    \item We perform Monte Carlo validation against pass-integrated Doppler FIMs, including a Starlink-representative Walker-delta shell across receiver latitudes and elevation masks. This demonstrates that the surrogate captures the dominant conditioning trends and yields constellation design implications.
\end{enumerate}

\subsection{Empirical Motivation}
\label{sec:motivation}
In this subsection, $\hat{\mathbf{n}}_i$ denotes the unit normal of the orbital plane of satellite $i$, and $\mathbf{v}_3$ the eigenvector of the smallest FIM eigenvalue; this and the related projection notation are formally introduced in Section~\ref{sec:orbital}. Before stating the model, we examine whether the pass-integrated Doppler Fisher information matrix exhibits the orbital-plane structure that this work exploits.  For a stationary receiver, the gradient of the Doppler observable with respect to receiver position is
\begin{equation}
\mathbf{g}(t)
= \frac{f_L}{c\,\rho(t)}
  \bigl[\mathbf{v}_{\mathrm{sat}}(t) - \dot{\rho}(t)\,\hat{\mathbf{r}}(t)\bigr],
\label{eq:grad}
\end{equation}
where $f_L$ is the carrier frequency, $c$ is the speed of light, $\rho(t) = \lVert \mathbf{r}_{\mathrm{sat}}(t) - \mathbf{r}_{\mathrm{rx}}\rVert$ is the satellite-to-receiver range, $\mathbf{v}_{\mathrm{sat}}(t)$ is the satellite velocity vector,
$\hat{\mathbf{r}}(t) = [\mathbf{r}_{\mathrm{sat}}(t) - \mathbf{r}_{\mathrm{rx}}]/\rho(t)$ is the unit line-of-sight vector from the receiver to the satellite, and $\dot{\rho}(t) = \mathbf{v}_{\mathrm{sat}}^T(t)\,\hat{\mathbf{r}}(t)$ is the range rate. Modeling the Doppler measurements as independent Gaussian observations with variance $\sigma_f^2$, the gradient \eqref{eq:grad} is the measurement Jacobian and the
pass-integrated FIM is
\begin{equation}
\mathbf{J}_{\mathrm{pass}}
= \frac{\Delta t}{\sigma_f^2}
  \sum_{k} \mathbf{g}(t_k)\mathbf{g}(t_k)^T,
\label{eq:Jpass}
\end{equation}
where $\Delta t$ is the sampling interval; the prefactor $\Delta t$ approximates the continuous-time information integral $\int \mathbf{g}(t)\mathbf{g}(t)^T\,dt/\sigma_f^2$ and does not affect the conditioning analysis below. The eigenstructure of $\mathbf{J}_{\mathrm{pass}}$ characterizes the directions in which a
satellite pass provides strong or weak positioning information.
\begin{table}[t]
\setlength{\tabcolsep}{4pt}
\renewcommand{\arraystretch}{0.95}
\centering
\caption{Eigenstructure of the pass-integrated Doppler FIM
$\mathbf{J}_{\mathrm{pass}}$ for four representative $550$\,km circular
passes ($f_L = 1.575$\,GHz, $\sigma_f = 1$\,Hz, $10^\circ$ elevation mask,
receiver at latitude $20^\circ$\,N), compared with the ideal projection
model $\gamma_i \mathbf{P}_i$. Near rank-deficiency along
$\hat{\mathbf{n}}_i$ is confirmed by $\lambda_3/\lambda_1$ and by the
alignment $|\mathbf{v}_3 \cdot \hat{\mathbf{n}}_i|$ of the minimum
eigenvector with the orbital-plane normal.}
\label{tab:motivation}
\begin{tabular}{lcccc}
\hline
Inclination & $\lambda_2/\lambda_1$ & $\lambda_3/\lambda_1$ &
$\hat{\mathbf{n}}_i^T \mathbf{J}_{\mathrm{pass}} \hat{\mathbf{n}}_i/\lambda_1$ &
$|\mathbf{v}_3 \cdot \hat{\mathbf{n}}_i|$ \\
\hline
$30^\circ$ & 0.390 & $1.1\times10^{-4}$ & 0.025 & 0.967 \\
$45^\circ$ & 0.393 & $3.6\times10^{-5}$ & 0.008 & 0.990 \\
$53^\circ$ & 0.393 & $2.3\times10^{-5}$ & 0.005 & 0.993 \\
$98^\circ$ & 0.388 & $1.6\times10^{-4}$ & 0.038 & 0.951 \\
\hline
$\gamma_i \mathbf{P}_i$ & 1 & 0 & 0 & 1 \\
\hline
\end{tabular}
\end{table}
Table~\ref{tab:motivation} reports the eigenstructure of
$\mathbf{J}_{\mathrm{pass}}$, evaluated via \eqref{eq:Jpass}, for four
representative 550\,km circular passes spanning inclinations from
$30^\circ$ to $98^\circ$, alongside that of the ideal surrogate
$\gamma_i \mathbf{P}_i$ with $\gamma_i =
\tfrac{1}{2}\operatorname{tr}(\mathbf{J}_{\mathrm{pass}})$. The
conditioning analysis scales uniformly with $f_L^2$ at fixed integration
time and SNR, so the structural conclusions are not specific to the L-band
reference frequency used here. In every case the third
eigenvalue is roughly four orders of magnitude smaller than the largest,
the normal-direction energy is $0.5$--$3.8\%$ of $\lambda_1$, and the
minimum eigenvector aligns with $\hat{\mathbf{n}}_i$ to within
$|\mathbf{v}_3 \cdot \hat{\mathbf{n}}_i| \ge 0.95$, all supporting the
projection structure. Two limitations are noted: the within-plane
eigenvalues are unequal ($\lambda_2/\lambda_1 \approx 0.39$, versus unity
for the surrogate), and the residual normal-direction energy is nonzero.
The model captures the dominant rank-deficiency structure while idealizing
the within-plane anisotropy. Section~II shows that this rank-deficiency
structure follows analytically from orbital dynamics
(Lemma~\ref{lem:normal}).

\section{Orbital Plane Information Model}
\label{sec:orbital}
The empirical rank deficiency observed in Section~\ref{sec:motivation} is not incidental;
it follows directly from orbital dynamics. The following elementary lemma
formalizes, within the measurement model~\eqref{eq:grad}, the long-standing qualitative
observation that Doppler measurements provide weak positioning information
transverse to the orbital plane~\cite{psiaki2021,levanon1998}.
 \begin{lemma}[Normal-direction suppression]\label{lem:normal}
Let the satellite follow a Keplerian orbit, whose plane contains the
geocenter, with unit normal $\hat{\mathbf{n}}$, and let
$d = \mathbf{r}_{\mathrm{rx}}^{T}\hat{\mathbf{n}}$ denote the receiver's
signed off-plane distance. Then the Doppler gradient~\eqref{eq:grad} satisfies
\begin{equation}\label{eq:lemma_normal}
  \mathbf{g}(t)^{T}\hat{\mathbf{n}}
  \;=\; \frac{f_L}{c}\,\frac{\dot{\rho}(t)\,d}{\rho^{2}(t)} .
\end{equation}
Moreover, $\|\mathbf{g}(t)\| = (f_L/c)\,v_{\perp}(t)/\rho(t)$, where
$v_{\perp}(t) = \|\mathbf{v}_{\mathrm{sat}}(t) -
\dot{\rho}(t)\hat{\mathbf{r}}(t)\|$ is the satellite velocity component
transverse to the line of sight, so that
$|\mathbf{g}(t)^{T}\hat{\mathbf{n}}| / \|\mathbf{g}(t)\|
= (|\dot{\rho}(t)|/v_{\perp}(t)) \, (d/\rho(t))$.
\end{lemma}
\begin{IEEEproof}
For a Keplerian orbit the plane contains the geocenter, so
$\mathbf{r}_{\mathrm{sat}}(t)^{T}\hat{\mathbf{n}} =
\mathbf{v}_{\mathrm{sat}}(t)^{T}\hat{\mathbf{n}} = 0$ for all $t$. Hence
\begin{equation*}
  \hat{\mathbf{r}}^{T}\hat{\mathbf{n}}
  = \frac{[\mathbf{r}_{\mathrm{sat}} -
     \mathbf{r}_{\mathrm{rx}}]^{T}\hat{\mathbf{n}}}{\rho}
  = -\frac{d}{\rho},
\end{equation*}
and substituting into~\eqref{eq:grad},
\begin{equation*}
  \mathbf{g}^{T}\hat{\mathbf{n}}
  = \frac{f_L}{c\rho}\!\left(\mathbf{v}_{\mathrm{sat}}^{T}\hat{\mathbf{n}}
    - \dot{\rho}\,\hat{\mathbf{r}}^{T}\hat{\mathbf{n}}\right)
  = \frac{f_L}{c}\,\frac{\dot{\rho}\,d}{\rho^{2}} .
\end{equation*}
The norm identity follows since
$\mathbf{v}_{\mathrm{sat}} - \dot{\rho}\hat{\mathbf{r}}$ is the component
of $\mathbf{v}_{\mathrm{sat}}$ orthogonal to $\hat{\mathbf{r}}$, and
$\dot{\rho} = \mathbf{v}_{\mathrm{sat}}^{T}\hat{\mathbf{r}}$.
\end{IEEEproof}
Since the FIM~\eqref{eq:Jpass} accrues as $\mathbf{g}(t)\mathbf{g}(t)^{T}$ at each
instant, the suppression of Lemma~\ref{lem:normal} holds pointwise in time
and therefore for the pass-integrated FIM over any observation window: the
normal-direction information is smaller than the in-plane information by
the pass-averaged square of $(|\dot{\rho}|/v_{\perp})(d/\rho)$. The
suppression is strongest for high-elevation passes, for which $d/\rho$
remains small throughout the pass, and weakens toward grazing geometries.
Because \eqref{eq:lemma_normal} is nonzero whenever $d \neq 0$, the lemma
also predicts the small but nonzero residual normal-direction energy
(0.5--3.8\%) observed in Table~\ref{tab:motivation} and its dependence on pass elevation.
Departures from ideal geometry (Earth-rotation receiver motion of
$\leq 0.47$\,km/s versus $\|\mathbf{v}_{\mathrm{sat}}\| \approx
7.6$\,km/s, small eccentricity, $J_2$ regression over a pass) perturb
\eqref{eq:lemma_normal} only weakly. Motivated by
Lemma~\ref{lem:normal} and validated empirically in Section~\ref{sec:motivation}, we adopt
the following geometric surrogate.
 
Let $\hat{\mathbf n}_i$ denote the unit normal vector of the orbital plane associated with satellite $i$. Define the projection matrix
\begin{equation}
\mathbf{P}_i = \mathbf{I}-\hat{\mathbf n}_i\hat{\mathbf n}_i^T .
\end{equation}
 The information contribution from satellite \(i\) is modeled as
\begin{equation}
\mathbf{J}_i =\gamma_i \mathbf{P}_i
\end{equation}
where $\gamma_i>0$ is a scalar information strength
representing the aggregate effect of Doppler signal quality,
carrier frequency, observation duration, and satellite-receiver
geometry. Among these factors, $\gamma_i$ scales as $f_L^2 /
\sigma_f^2$ from \eqref{eq:Jpass}, so carrier frequency is the
dominant physical lever; integration time and SNR enter
linearly through the number of measurements and $\sigma_f^2$,
respectively. The in-plane geometric factor depends on the
receiver-satellite geometry over the pass and is not easily
decomposed further without reference to a specific orbit, but
in the Monte Carlo evaluation of Section~\ref{sec:numerical}
it is absorbed into $\gamma_i = \tfrac{1}{2}\operatorname{tr}(\mathbf{J}_{\mathrm{pass},i})$. The isotropic in-plane weighting of the surrogate (equal information along both in-plane axes) is an idealization; its impact on positioning accuracy is quantified in Section~\ref{sec:discussion}.

For $N$ satellites the total information matrix is 
\begin{equation}
\mathbf{J} = \textstyle \sum_{i=1}^{N} \gamma_i \mathbf{P}_i = \sum_{i=1}^{N} \gamma_i
\bigl(\mathbf{I} - \hat{\mathbf{n}}_i \hat{\mathbf{n}}_i^T\bigr).
\label{eq:JtotalN}
\end{equation}
The combined matrix $\mathbf{J}$ is rank-deficient along any
direction orthogonal to all orbital planes; for $N \ge 3$
non-coplanar orbital normals such a direction generically does not
exist and $\mathbf{J}$ is full rank. The two-satellite case is
therefore the critical configuration: it isolates the pairwise
conditioning effect governed by the dihedral angle between planes,
and the closed-form analysis of Section~\ref{sec:eigen} carries
over to larger constellations as a pairwise lower bound
(Section~\ref{sec:nsat}).
 
For two satellites the total information matrix reduces to
 \begin{equation}
\mathbf{J}
=
\gamma_1 \mathbf{P}_1+\gamma_2 \mathbf{P}_2 .
\label{eq:Jtotal}
\end{equation}
Under this normalization $\mathbf{J}$ carries units of position Fisher
information (m$^{-2}$) and acts directly as a position FIM, so
worst-case position-variance bounds follow from $1/\lambda_{\min}(\mathbf{J})$.

\section{Eigenstructure of the Orbital Plane Information Model}
\label{sec:eigen}
\subsection{Two-Orbital-Plane Geometry}
Consider the two-satellite model of \eqref{eq:Jtotal}.
Let $\phi \in [0^\circ,90^\circ]$ denote the acute dihedral
angle between orbital planes $\mathcal P_1$ and $\mathcal P_2$,
defined as $\phi = \arccos(|\hat{\mathbf{n}}_1 \cdot
\hat{\mathbf{n}}_2|)$ so that $\phi = 0^\circ$ for coplanar
orbits and $\phi = 90^\circ$ for orthogonal orbital planes.
For an arbitrary angle $\phi' \in [0^\circ,180^\circ]$ between
plane normals, conditioning depends only on the acute reduction
$\phi = \min(\phi', 180^\circ - \phi')$, since the discriminant
in \eqref{eq:eigenvalues} depends only on $\cos^2\phi$. The
following theorem characterizes the eigenstructure.
\begin{theorem}[Eigenstructure of Two Orbital Planes]
\label{thm:condnum}
The eigenvalues of 
$\mathbf{J}$ are:
\begin{multline}
\lambda_{1,2} = \frac{(\gamma_1+\gamma_2) \pm 
\sqrt{(\gamma_1-\gamma_2)^2 + 4\gamma_1\gamma_2\cos^2\phi}}{2}, 
\\ \lambda_3 = \gamma_1 + \gamma_2
\label{eq:eigenvalues}
\end{multline}
For equal information levels $\gamma_1 = \gamma_2 = \gamma$:
\begin{equation}
\lambda_{\min} = \gamma(1 - \cos\phi) = 2 \gamma \sin^2(\phi/2)
\label{eq:lambdamin}
\end{equation}
giving condition number:
\begin{equation}
\kappa = \frac{\lambda_{\max}}{\lambda_{\min}} = 
\frac{1}{\sin^2(\phi/2)}
\label{eq:condnum}
\end{equation}
\end{theorem}
\begin{proof}
Choose coordinates with $\hat{\mathbf{n}}_1 = 
[0,0,1]^T$ and $\hat{\mathbf{n}}_2 = 
[\sin\phi, 0, \cos\phi]^T$ for 
$\phi \in [0^\circ, 90^\circ]$. Substituting 
into~\eqref{eq:Jtotal}:
\begin{equation}
\mathbf{J} = 
\begin{bmatrix} 
\gamma_1 + \gamma_2\cos^2\phi & 0 & -\gamma_2\sin\phi\cos\phi \\ 
0 & \gamma_1+\gamma_2 & 0 \\ 
-\gamma_2\sin\phi\cos\phi & 0 & \gamma_2\sin^2\phi 
\end{bmatrix}
\label{eq:matrix}
\end{equation}
The $y$-direction decouples with eigenvalue 
$\lambda_3 = \gamma_1 + \gamma_2$. The characteristic equation 
of the $x$-$z$ block is:
\begin{equation}
\lambda^2 - \lambda(\gamma_1+\gamma_2) + \gamma_1 \gamma_2\sin^2\phi = 0
\label{eq:charpoly}
\end{equation}
giving $\lambda_{1,2}$ as in~\eqref{eq:eigenvalues}. 
For $\gamma_1 = \gamma_2 = \gamma$, the discriminant reduces to 
$4 \gamma ^2\cos^2\phi$, giving 
$\lambda_{\min} = \gamma(1-\cos\phi)$. For $\phi>0$ the largest eigenvalue is
$\lambda_{\max}=2\gamma$, so
$\kappa = 2\gamma/[\gamma(1-\cos\phi)]
= 1/\sin^2(\phi/2)$,
which gives~\eqref{eq:condnum}.
\end{proof}
\begin{corollary}[Worst-Case Accuracy]
\label{cor:design}
The smallest eigenvalue of $\mathbf{J}$ is
\begin{equation}
\lambda_{\min}
=
\frac{
(\gamma_1+\gamma_2)
-
\sqrt{
(\gamma_1-\gamma_2)^2
+
4\gamma_1\gamma_2\cos^2\phi
}
}{2}.
\end{equation}
The worst-case CRLB is therefore
\begin{equation}
\sigma_{\text{worst}}^2
\ge
\frac{1}
{\lambda_{\min}}.
\label{eq:worst}
\end{equation}
For small $\phi$,
\begin{equation}
\lambda_{\min}
\approx
\frac{\gamma_1\gamma_2}
{\gamma_1+\gamma_2}
\phi^2,
\end{equation}
giving
\begin{equation}
\sigma_{\text{worst}}^2
\approx
\frac{\gamma_1+\gamma_2}
{\gamma_1\gamma_2\phi^2}.
\end{equation}
\end{corollary}
For fixed total information $S = \gamma_1 + \gamma_2$, the
small-angle bound \eqref{eq:worst} is minimized at $\gamma_1 =
\gamma_2 = S/2$ by the arithmetic--geometric mean inequality,
identifying link-budget balance as a design lever independent
of geometry.
\subsection{Extension to $N$ Satellites}
\label{sec:nsat}
The two-satellite eigenstructure of Theorem~\ref{thm:condnum}
extends to a conditioning guarantee for arbitrary $N$. While the
general $N$-satellite eigenvalues admit no comparably simple
closed form, the worst-case conditioning is controlled by the
best-conditioned pair.
\begin{proposition}[Pairwise Conditioning Bound]
\label{prop:pairwise}
For the $N$-satellite information matrix $\mathbf{J} = \sum_{i=1}^N
\gamma_i \mathbf{P}_i$ with $\gamma_i > 0$,
\begin{equation}
\lambda_{\min}(\mathbf{J}) \;\ge\; \max_{i<j}\,
\lambda_{\min}\!\left(\gamma_i \mathbf{P}_i + \gamma_j \mathbf{P}_j\right),
\label{eq:pairwise}
\end{equation}
where each pairwise term is given in closed form by
Corollary~\ref{cor:design} with dihedral angle $\phi_{ij} =
\arccos(|\hat{\mathbf{n}}_i \cdot \hat{\mathbf{n}}_j|)$.
Consequently, the worst-case CRLB satisfies
\begin{equation}
\sigma^2_{\mathrm{worst}} \;\le\;
\min_{i<j}\, \frac{2}{(\gamma_i+\gamma_j) -
\sqrt{(\gamma_i-\gamma_j)^2 + 4\gamma_i\gamma_j\cos^2\phi_{ij}}}.
\label{eq:pairwise_crlb}
\end{equation}
\end{proposition}
\begin{proof}
Fix any pair $(i,j)$ and write $\mathbf{J} = \mathbf{J}_{ij} + \mathbf{R}$ with $\mathbf{J}_{ij} =
\gamma_i \mathbf{P}_i + \gamma_j \mathbf{P}_j$ and $\mathbf{R} = \sum_{k \neq i,j} \gamma_k
\mathbf{P}_k$. Each $\mathbf{P}_k$ is an orthogonal projection with eigenvalues
$\{1,1,0\}$, so $\mathbf{P}_k \succeq 0$ and hence $\mathbf{R} \succeq 0$. For any
unit vector $\hat{\mathbf{u}}$, $\hat{\mathbf{u}}^T \mathbf{J}
\hat{\mathbf{u}} = \hat{\mathbf{u}}^T \mathbf{J}_{ij} \hat{\mathbf{u}} +
\hat{\mathbf{u}}^T \mathbf{R} \hat{\mathbf{u}} \ge \hat{\mathbf{u}}^T
\mathbf{J}_{ij} \hat{\mathbf{u}}$. Minimizing both sides over
$\hat{\mathbf{u}}$ and applying the Courant--Fischer theorem
gives $\lambda_{\min}(\mathbf{J}) \ge \lambda_{\min}(\mathbf{J}_{ij})$. Since this
holds for every pair, \eqref{eq:pairwise} follows;
\eqref{eq:pairwise_crlb} substitutes the closed-form
$\lambda_{\min}(\mathbf{J}_{ij})$ of Corollary~\ref{cor:design}.
\end{proof}
\begin{figure*}[t]
\centering
\includegraphics[scale=0.43]{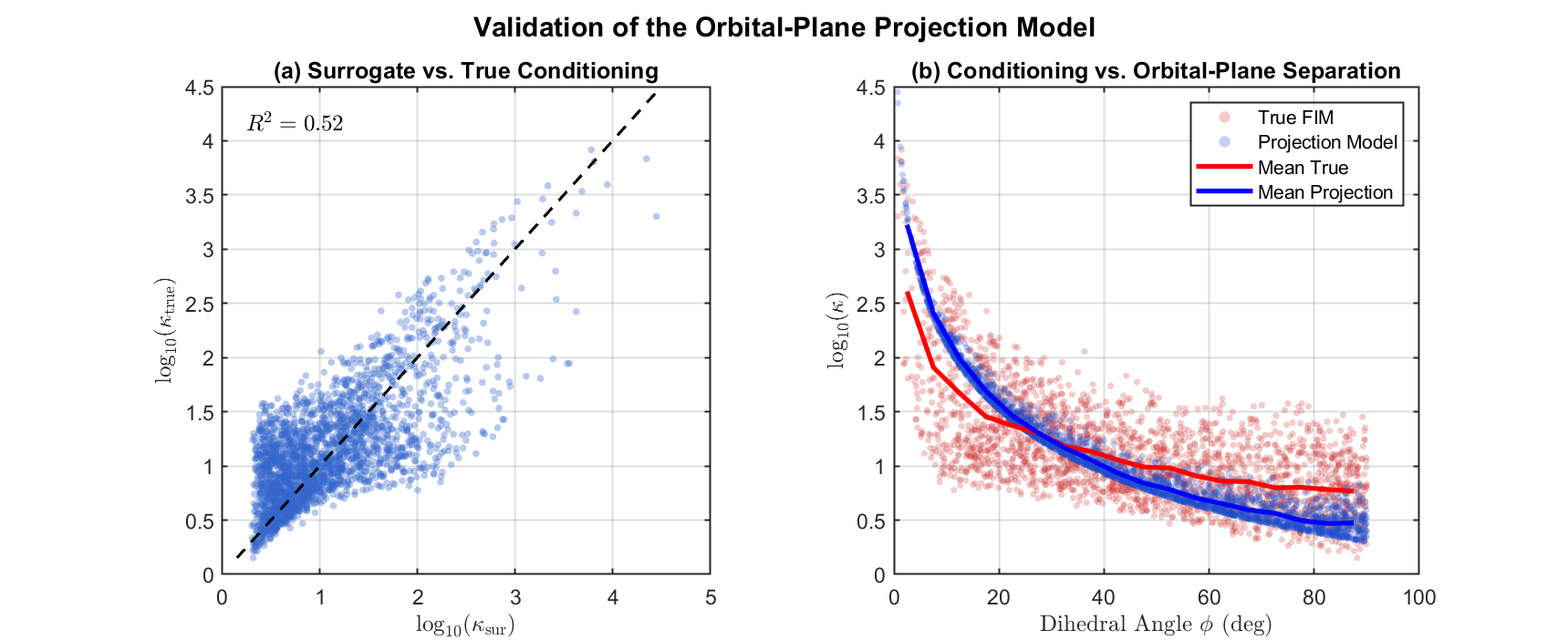}
\caption{Validation of the orbital-plane projection model.
(a) Comparison between the condition number predicted by the
projection surrogate and that obtained from the true
pass-integrated Doppler FIM. The dashed line indicates
perfect agreement. (b) Conditioning versus orbital-plane
dihedral angle. Both the true Doppler FIM and the
projection surrogate exhibit rapidly degrading conditioning
as orbital planes become coplanar.}
\label{fig:surrogate_validation}
\end{figure*}
For $N=2$ the residual $\mathbf{R}$ is empty and \eqref{eq:pairwise}
holds with equality, recovering Theorem~\ref{thm:condnum}. For
equal information $\gamma_i = \gamma$, the bound reduces to
$\lambda_{\min}(\mathbf{J}) \ge \gamma\,(1 - \cos\phi_{\max})$, where
$\phi_{\max} = \max_{i<j}\phi_{ij}$ is the largest pairwise
plane separation. The worst-case conditioning of a constellation
is therefore no worse than that of its best-separated, best-balanced
pair: a single well-conditioned pair suffices to guarantee the
whole constellation's accuracy floor, and additional satellites
can only improve it. The bound is tightest in the ill-conditioned
regime of interest, since pairwise near-coplanarity across all
pairs forces the normals to cluster and precludes collective
spanning of the cross-track direction.
\section{Numerical Validation}
\label{sec:numerical}
\label{se:validC}
To assess whether orbital-plane geometry predicts the
conditioning of practical Doppler information matrices, we perform a Monte Carlo study. Each trial generates two independent circular 550\,km orbits with
random inclination, right ascension of the ascending node (RAAN), and
orbital phase. The pass-integrated Doppler FIM of each pass is
computed via \eqref{eq:Jpass} with a $10^\circ$ elevation mask, the
two are summed to form the true two-satellite FIM, and the
corresponding surrogate is constructed from \eqref{eq:Jtotal} with
$\gamma_i=\tfrac{1}{2}\operatorname{tr}(\mathbf{J}_i)$. Condition
numbers were evaluated for both models over 2500 trials.
 
Fig.~\ref{fig:surrogate_validation}(a) compares the
condition number predicted by the projection surrogate
with that of the true Doppler FIM. The surrogate captures a substantial fraction of the observed variation in conditioning,
with correlation coefficient $\rho=0.72$
($R^2=0.52$) between
$\log_{10}(\kappa_{\mathrm{sur}})$ and
$\log_{10}(\kappa_{\mathrm{true}})$. The moderate $R^2$
reflects the isotropic in-plane weighting of the surrogate
rather than a failure of the orbital-plane projection
structure, since the within-plane anisotropy
($\lambda_2/\lambda_1 \approx 0.27$) is the dominant source
of approximation error as discussed in Section~\ref{sec:discussion}.
 
Fig.~\ref{fig:surrogate_validation}(b) shows the dependence
of conditioning on orbital-plane separation. Both the
surrogate and the true Doppler FIM exhibit rapidly
increasing condition number as the dihedral angle
decreases, confirming that orbital-plane diversity is
a dominant geometric factor governing Doppler-only
positioning performance. The $0.27$ ensemble mean of $\lambda_{2}/\lambda_1$ is lower than the $\approx 0.39$ of
Section~\ref{sec:motivation} because the random orbital sampling
includes grazing-elevation passes, which exhibit weaker in-plane
diversity than the high-elevation passes selected there. The
primary approximation error is therefore the isotropic in-plane
weighting, not residual normal-direction information.
 
To assess sensitivity to pass conditions, the ensemble was
re-evaluated with observation windows truncated from $\pm 900$\,s to
$\pm 450$\,s and $\pm 225$\,s about a random epoch, and with elevation
masks of $5^\circ$, $10^\circ$, and $20^\circ$. The projection structure
is insensitive to these choices: the smallest FIM eigenvalue remains below
$0.1\%$ of the trace in every condition, the surrogate correlation is
stable ($R^2 = 0.52$--$0.56$ across window lengths), and only the
within-plane ratio varies mildly
($\bar{r} = \overline{\lambda_2/\lambda_1} = 0.22$--$0.27$), consistent
with the window-invariant suppression of Lemma~\ref{lem:normal}.
 
\textit{Walker-shell validation.} To evaluate a realistic
constellation, the random ensemble was replaced by a Walker-delta shell
representative of the Starlink first shell ($550$\,km, $i = 53^\circ$, 72
planes). Because passes are collected at different times, Earth rotation
renders the node-to-receiver longitude offset uniformly distributed per
pass; each trial therefore draws two visible passes from the shell at
independent epochs. For a common-inclination shell the pairwise plane
geometry admits the closed form
\begin{equation}
\cos\varphi' \;=\; \cos^{2} i + \sin^{2} i\,\cos\Delta\Omega,
\label{eq:walker}
\end{equation}
where $\Delta\Omega$ is the node separation, so the acute
dihedral angle is capped at $\min(2i, 90^\circ)$. By
Theorem~\ref{thm:condnum}, a shell with $i < 45^\circ$ therefore incurs a
two-satellite conditioning floor of $\kappa \ge 1/\sin^{2} i$ regardless of
the number of planes; the floor holds for any information levels, since at fixed total information $S$, $\lambda_{\min} = \bigl[S - (S^{2} - 4\gamma_1\gamma_2\sin^{2}\phi)^{1/2}\bigr]/2$ is maximized at $\gamma_1 = \gamma_2$. For $i = 53^\circ$ the cap is inactive, but \eqref{eq:walker} concentrates the induced dihedral distribution at small plane separations.
 
Table~\ref{tab:walker} reports the surrogate-to-true
conditioning correlation, the median true condition number, and the mean
within-plane ratio across receiver latitudes and elevation masks (600
trials per condition). Three trends emerge. First, the surrogate tracks
single-shell conditioning closely, with $R^2$ increasing from $0.59$ to
$0.97$ toward the ill-conditioned high-latitude regime, consistent with
the tightness of the pairwise bound noted in Section~\ref{sec:nsat}.
Second, near-coplanar pass pairs ($\varphi < 20^\circ$) constitute
$29$--$50\%$ of visible pairs at latitudes up to $45^\circ$, compared to
$16$--$21\%$ for the inclination-diverse ensemble. This demonstrates the
concentration of \eqref{eq:walker} near $\Delta\Omega \approx 0$. Third,
conditioning collapses as receiver latitude approaches the shell
inclination: at $60^\circ$ latitude the median condition number rises to
$\approx 70$--$80$ under $5^\circ$--$10^\circ$ masks and to $\approx 590$
under a $20^\circ$ mask, where $96\%$ of pairs are near-coplanar and
$\bar{r}$ falls to $0.16$; the corresponding median for the
inclination-diverse ensemble at the same location is $\approx 8$. For the
inclination-diverse random ensemble the statistics of
Table~\ref{tab:walker} are stable across all conditions
($R^2 = 0.45$--$0.73$, median $\kappa = 8$--$16$,
$\bar{r} = 0.26$--$0.28$) and are therefore not tabulated. Design
implications are discussed in Section~\ref{sec:discussion}.
 
\begin{table}[t]
\centering
\caption{Walker-shell validation ($550$\,km,
$53^\circ$, 72 planes) vs. receiver latitude and elevation mask:
$R^{2}$ between $\log_{10}\kappa_{\mathrm{sur}}$ and
$\log_{10}\kappa_{\mathrm{true}}$ / median $\kappa_{\mathrm{true}}$ /
mean $\bar{r}=\overline{\lambda_2/\lambda_1}$ (600 trials per
condition).}
\label{tab:walker}
\begin{tabular}{lccc}
\hline
 & \multicolumn{3}{c}{Elevation mask \;($R^{2}$ / med.\ $\kappa$ / $\bar{r}$)} \\
Latitude & $5^\circ$ & $10^\circ$ & $20^\circ$ \\
\hline
$0^\circ$  & 0.59 / 16 / 0.26 & 0.72 / 14 / 0.27 & 0.86 / 9 / 0.27 \\
$20^\circ$ & 0.63 / 17 / 0.26 & 0.74 / 14 / 0.27 & 0.86 / 9 / 0.27 \\
$45^\circ$ & 0.70 / 13 / 0.30 & 0.78 / 12 / 0.30 & 0.82 / 14 / 0.25 \\
$60^\circ$ & 0.88 / 68 / 0.24 & 0.93 / 80 / 0.23 & 0.97 / 592 / 0.16 \\
\hline
\end{tabular}
\end{table}

\section{Discussion}
\label{sec:discussion}
 \textit{Practical implications.} 
Corollary~\ref{cor:design} gives direct design rules for 
LEO SoOP constellation geometry. The dihedral angle 
$\phi$ is the primary geometric lever: maximizing $\phi$ minimizes the condition number and improves the isotropy of the information matrix. The 
information balance condition $\gamma_1 \approx \gamma_2$ 
identifies link budget as a secondary and independent 
design parameter. The model thereby gives a quantitative
geometric account of the well-known deterioration of cross-track
accuracy under near-coplanar satellite geometries.
 
\textit{Single-shell constellations.}
For positioning against a single Walker shell, \eqref{eq:walker} makes
inclination the governing design parameter: shells with $i < 45^\circ$ cap
the pairwise dihedral angle at $2i$ and hence the achievable conditioning
at $\kappa \ge 1/\sin^{2} i$. Even for $i > 45^\circ$ the induced
dihedral distribution concentrates near coplanarity, and visibility
filtering intensifies this concentration as receiver latitude approaches
the shell inclination (Table~\ref{tab:walker}). Robust Doppler-only
positioning at latitudes near or above a shell's inclination therefore
benefits disproportionately from inclination diversity---combining shells
of different inclinations or augmenting with polar orbits---rather than
from additional planes within a single shell.
 
\textit{Receiver location and latitude dependence.}
The coupling between receiver location, visibility
constraints, and information geometry is quantified in
Table~\ref{tab:walker}. As Lemma~\ref{lem:normal} predicts, the near
rank-deficiency along $\hat{\mathbf{n}}_i$ follows from orbital dynamics
and persists at all latitudes and masks: the smallest FIM eigenvalue
remains below $0.1\%$ of the trace in every condition. By contrast, within-plane anisotropy varies with visibility geometry: $\bar{r}$ is
essentially latitude-invariant for the inclination-diverse ensemble
($0.26$--$0.28$) but falls to $0.16$ for the single Walker shell at
$60^\circ$ latitude under a $20^\circ$ mask, where only grazing, nearly
parallel pass geometries remain visible. This dependence of the observable
orbital geometry on station latitude and minimum elevation angle parallels
the ground-station-centric Doppler analysis of~\cite{luxu2026}.
 
\textit{Limitations.}
While Lemma~\ref{lem:normal} grounds the rank-deficiency
structure in orbital dynamics, the isotropic in-plane weighting remains an
idealization: the two dominant in-plane eigenvalues are generally unequal,
with ensemble mean $r := \lambda_2/\lambda_1 \approx 0.27$
($0.16$--$0.30$ across the conditions of
Table~\ref{tab:walker}). This anisotropy is the principal source of
surrogate error. Trace-matching to the true in-plane eigenvalues gives a
small-angle variance optimism of $(1+r)/2r \approx 2.4$, so the surrogate underestimates the worst-case position
standard deviation by $\sqrt{2.4} \approx 1.6\times$. More consequentially,
anisotropy imposes a conditioning floor that persists at orthogonal planes:
when the weak in-plane axis aligns with the cross-nodal direction,
$\kappa \approx 2/r$ at $\phi = 90^\circ$ rather than unity, so plane
orthogonality is necessary but not sufficient for well-conditioned
positioning. Replacing the scalar $\gamma_i$ with a $2\times2$ in-plane
covariance, while preserving tractability, remains future work.

\section{Conclusion}
We presented an idealized orbital-plane information model for Doppler-only positioning with LEO signals of opportunity, grounded in Keplerian dynamics by Lemma~\ref{lem:normal}. Representing each satellite contribution as a weighted projection onto its orbital plane yields closed-form expressions for the eigenvalues, condition number, and worst-case positioning accuracy, identifying orbital-plane separation and information balance as independent design factors. Monte Carlo evaluation against pass-integrated Doppler Fisher information matrices, including a Starlink-like Walker shell across receiver latitudes and
elevation masks, confirms that the surrogate captures the dominant conditioning trends driven by orbital-plane geometry. The principal approximation error is within-plane anisotropy, which imposes a conditioning floor that orbital-plane orthogonality alone cannot overcome.


\begin{thebibliography}{9}
\bibitem{kassas2019}
Z.~M. Kassas, J.~Morales, and J.~J. Khalife,
``New-age satellite-based navigation---STAN:
Simultaneous tracking and navigation with LEO
satellite signals,'' \emph{Inside GNSS}, vol.~14,
no.~4, pp.~56--65, Jul./Aug. 2019.
 
\bibitem{psiaki2021}
M.~L. Psiaki, ``Navigation using carrier Doppler
shift from a LEO constellation: TRANSIT on
steroids,'' \emph{NAVIGATION}, vol.~68, no.~3,
pp.~621--641, Sep. 2021.
 
\bibitem{khalife2019}
J.~J. Khalife and Z.~M. Kassas,
``Receiver design for Doppler positioning with
LEO satellites,'' in \emph{Proc. IEEE ICASSP},
Brighton, UK, May 2019, pp.~5506--5510.
 
\bibitem{levanon1998}
N. Levanon, ``Quick position determination using 1 or 2 LEO satellites," \emph{IEEE Trans. Aerosp. Electron. Syst.}, vol.~34, no.~3,
pp.~736--754, Jul. 1998.
 
\bibitem{moralesferre2020}
R.~Morales-Ferre, E.~S. Lohan, G.~Falco, and E.~Falletti,
``GDOP-based analysis of suitability of LEO constellations
for future satellite-based positioning,''
in \emph{Proc. IEEE WiSEE}, Vicenza, Italy, Oct. 2020,
pp.~147--152.
 
\bibitem{mclemore2021}
B.~McLemore and M.~L. Psiaki,
``GDOP of navigation using pseudorange and Doppler shift
from a LEO constellation,''
in \emph{Proc. ION GNSS+}, St. Louis, MO, Sep. 2021,
pp.~2471--2492.
 
\bibitem{emenonye2025}
D.-R. Emenonye, H.~S. Dhillon, and R.~M. Buehrer,
``Fundamentals of LEO-based localization,''
\emph{IEEE Trans. Inf. Theory}, vol.~71, no.~7,
pp.~5277--5311, Jul. 2025.
 
\bibitem{luxu2026}
S.~Lu and G.~Xu,
``Doppler shift analysis in LEO satellite systems:
A ground-station perspective,''
\emph{IEEE Trans. Veh. Technol.}, early access, 2026,
doi: 10.1109/TVT.2026.3701255.
 
\bibitem{neinavaie2021}
M.~Neinavaie, J.~Khalife, and Z.M.~Kassas, ``Acquisition, Doppler Tracking, and Positioning with Starlink LEO Satellites: First Results,'' \emph{IEEE Trans. Aerosp. Electron. Syst.}, vol.~58, 
no.~3,
pp.~2606--2610, Nov. 2021.
 
\end{thebibliography}
\end{document}